\documentclass[10pt,a4paper,twoside]{article}

\usepackage[latin1]{inputenc}
\usepackage{amssymb}
\usepackage{amsmath}
\usepackage{amsfonts}
\usepackage{amsbsy}
\usepackage{natbib}
\usepackage{eucal}
\usepackage{graphicx}
\usepackage{multirow}
\usepackage[english]{babel}
\usepackage{mathptmx}
\usepackage{float}
\usepackage[font=small]{caption}
\usepackage{array}
\usepackage[colorlinks,citecolor=blue,urlcolor=blue]{hyperref}
\usepackage{amsthm}
\usepackage{comment}
\usepackage{setspace}
\usepackage{bm}
\usepackage{arxiv}
\newtheorem{theorem}{Theorem}[section]

\newtheorem{remark}{Remark}[section]

\begin{document}
\title{New characterization based exponentiality  tests for randomly censored data}

\author{Marija Cupari\'c\thanks{marijar@matf.bg.ac.rs},  Bojana Milo\v sevi\'c\thanks{bojana@matf.bg.ac.rs } 
         \\\medskip
{\small Faculty of Mathematics, University of Belgrade, Studenski trg 16, Belgrade, Serbia}}

\date{}

\maketitle
\begin{abstract}
Recently, the characterization based approach for the construction of goodness of fit tests has become popular. Most of the proposed tests have been designed for complete i.i.d. samples. 
Here we present the adaptation of the recently proposed exponentiality tests based on equidistribution-type characterizations for the case of randomly censored data. Their asymptotic properties are provided. Besides, we present the results of wide empirical power study including the powers of several recent competitors. This study can be used as a benchmark for future tests proposed for this kind of data.
\end{abstract}

\section{Introduction}\label{intro}
The exponential distribution is one of the simplest lifetime distribution, and, as such,  the most used one. Hence, the construction of a battery of powerful exponentiality tests is of huge importance.   For the case of i.i.d. data there have been several extensive studies with this aim (see e.g. \cite{henze2005}, \cite{torabi2018wide}, \cite{jimenez2020exponentiality}).
Step forward is also made in adapting common exponentiality tests for testing distributional assumptions in some time-series model such as conditional duration models (see \cite{meintanis2017goodness}).  All mentioned studies have assumed the complete data case.

However, in survival analysis, our sample is often limited due to some censoring mechanism. Here, the outcome variable of interest is the time until an event occurs (survival time). Most of the times, this type of data is censored, that means
we have information about the survival time of the individual, but we don't know exactly
the survival time. 
This situation often happens in health studies where it is not possible to observe all individuals until the time of death or some individuals decide to drop off from the study. This is a typical example of right censoring. 
On the other hand, the limitations of measurement devices often lead to left censorship.  Sometimes, both situations are present. This is so-called "double censorship". In all three mentioned cases there are three possible censoring mechanisms: all sample elements greater (lower) than some constant value are not noted (type I censoring), the censoring levels are random variables (random censoring), a fixed number of greater (lower) sample elements are not noted (type II censoring). Some hybrid versions of previously mentioned mechanisms are also possible.  
Notice that type I censoring is a particular case of the random censorship when the censoring variable is degenerate, hence it is more natural to study later censorship scheme. Moreover, although type I naturally arise as a consequence of time limitation of a clinical study, the individuals enter the study in more or less random fashion after diagnosis of some disease of interest, so their survival times are randomly censored (see e.g. \cite{lawless2002statistical}).

 Due to all mentioned above, there is a need for studying the goodness-of-fit tests for this kind of data.   
Although among the first tests for randomly censored data have been proposed in \cite{barlow1969note}, \cite{koziol1976cramer}, 
until now this field of research is still not well explored. 
 In \cite{koziol1976cramer}, the Cram\'er-von Mises statistic is generalised to randomly censored data using the Kaplan-Meier product-limit estimator of the distribution function. This approach might be used for adaptation of all test statistics that are based on some difference between empirical and hypothesised distribution function, or are some functionals of distribution function (see e.g. \cite{strzalkowska2017goodness}).  The null distributions of such tests usually depend on censoring distribution and are not scale-location free. Hence,  some resampling procedures are typically required. 
 One possible approach to cope with this issue is to apply tests on pseudo-complete samples (see \cite{balakrishnan2015empirical}). The other method, possible for large sample sizes, when the limiting distribution of test statistic is normal, is to consistently estimate asymptotic variance and use quantiles of the standard normal distribution for determination of critical region.

The aim of this paper is twofold.  First,  to
 adapt recent exponentiality tests based on U-empirical Laplace transforms for the case of randomly censored data and explore their properties. For the adaptation, we use so-called "inverse probability censoring weights" approach (IPCW) (see \cite{robins1992recovery}), where each unit is weighted by the inverse of an estimate of the conditional
probability of having remained uncensored.
Second, since there is no comparative study known so far, we conduct the
extensive power study for a large number of commonly used alternatives, through which we demonstrate the sensitivity of tests to data incompleteness.  Moreover, the presented results can serve as a benchmark for future tests for randomly censored data.

The paper is organized as follows.
In Section \ref{sec: test stat} we introduce test statistics and derive their asymptotic null distributions under random censorship. Section \ref{sec: empirical} contains  empirical study and discussion.  All proofs are given in the Appendix.
\section{Test statistics and their asymptotic properties}\label{sec: test stat}
 Let ${X}'_1,...,{X}'_n$ be a random sample from a non-negative continuous distribution function $F(x)$.
 For  testing null hypothesis $H_0:F(x)=1-e^{-\frac{x}{\mu}},\;x>0,$ for some particular $\mu>0$,  several classes of test statistics based on V- or U- empirical Laplace transforms were proposed in \cite{MilosevicObradovicPapers}, \cite{cuparic2019new} and \cite{cuparic2018new}. 
 Here we consider appropriate versions via U-statistics approach, given below
\begin{equation}\label{integralna}
    J^\mathcal{I}_{n,a}=\int\limits_0^{\infty}\Big(\frac{1}{n}\sum_{i_1=1}^ne^{-tX'_{i_1}}-\frac{1}{\binom{n}{2}}\sum_{1\leq i_1<i_2\leq n}e^{-t \psi^\mathcal{I}(X'_{i_1},X'_{i_2})}\Big)e^{-at}dt,
\end{equation}
\begin{equation}\label{L2}
    M^{\mathcal{I}}_{n,a}=\int\limits_0^{\infty}\Big(\frac{1}{n}\sum_{i_1=1}^ne^{-tX'_{i_1}}-\frac{1}{\binom{n}{2}}\sum_{1\leq i_1<i_2\leq n}e^{-t \psi^\mathcal{I}(X'_{i_1},X'_{i_2})}\Big)^2e^{-at}dt,
\end{equation}
where 
$\mathcal{I}$ indicates the characterization  the test is based on ($\mathcal{P}$ for Puri-Rubin, and $\mathcal{D}$ for Desu characterization).
We highlight that, to be suitable for testing a composite hypothesis,  those tests in \cite{cuparic2018new} and \cite{cuparic2019new} were originally applied to the scaled sample.

It can be noticed that statistics of the form  \eqref{integralna} are $U-$statistic with kernel
\begin{equation}\label{jezgroPhi}
    \Phi^\mathcal{I}(x_1,x_2;a)=\frac{1}{2}\bigg(\frac{1}{a+x_1}+\frac{1}{a+x_2}-\frac{2}{a+\psi^\mathcal{I}(x_1,x_2)}\bigg),
\end{equation}
while the statistic \eqref{L2} could be expressed as 
\begin{equation*}
    M^{\mathcal{I}}_{n,a}=\int\limits_0^\infty U^2_n(t)e^{-at}dt,
\end{equation*}
where $\{U_n(t)\}$ is a $U$-empirical process   of degree two with kernel 
\begin{equation}\label{jezgroh}
    h^\mathcal{I}(x_1,x_2;t)=\frac{1}{2}\Big(e^{-tx_{1}}+e^{-tx_{2}}-2e^{-t \psi^\mathcal{I}(x_{1},x_{2})}\Big).
\end{equation}
 In the expressions above 
$\psi^\mathcal{P}(x,y)=|x-y|,\; \psi^\mathcal{D}(x,y)=2\min(x,y). $


In what follows, we modify those test statistics to be applicable in case of randomly censored data.
Let, again, $X'_1,...,X'_n$ be a random sample from a non-negative continuous distribution function $F(x)$.
Let $C_1,...,C_n$ be censoring random variables with absolutely continuous  distribution function $G$ defined on $\mathbb{R}^+$. In the context of right censored data, for every $i=1,...,n$, we observe $X_i=\min\{X'_i,C_i\}$ and $\delta_i=I(X'_i\leq C_i)$. We assume the independent censoring model. 
Following  \cite{datta2010inverse}, we employ IPCW approach to make suitable modifications of \eqref{integralna} and \eqref{L2}. In particular, for $I\in\{\mathcal{P},\mathcal{D}\}$, we consider statistics
\begin{equation}\label{integralnaC}
    J^\mathcal{I}_{c,a}=\frac{1}{\binom{n}{2}}\sum_{i<j}\frac{\Phi^\mathcal{I}(X_i,X_j;a)\delta_i\delta_j}{K_c(X_i-)K_c(X_j-)},
\end{equation}
\begin{equation}\label{L2C}
    M^\mathcal{I}_{c,a}=\int\limits_0^\infty (U^\mathcal{I}_c(t))^2e^{-at}dt=\int\limits_0^\infty\bigg(\frac{1}{\binom{n}{2}}\sum_{i<j}\frac{h^\mathcal{I}(X_i,X_j;t)\delta_i\delta_j}{K_c(X_i-)K_c(X_j-)}\bigg)^2e^{-at}dt,
\end{equation} 
where $\Phi^\mathcal{I}$ and $h^\mathcal{I}$ are given earlier, and $K_c(x-)=1-G(x-)=P\{C\geq x\}$ is the survival function of the censoring variable $C$ such that $K_c(x)>0$, for each $x$, with probability 1.   

If the survival distribution of censoring variable is unknown, which usually occurs in practice,  $K_c$ have to be replaced by its consistent estimator. Here, we use the Kaplan-Meier estimator, where the role of censored and failed observations are reversed. Then statistics \eqref{integralnaC} and \eqref{L2C} become 
\begin{equation}\label{integralnaC1}
    \widehat{J}^\mathcal{I}_{c,a}=\frac{1}{\binom{n}{2}}\sum_{i<j}\frac{\Phi^\mathcal{I}(X_i,X_j;a)\delta_i\delta_j}{\widehat{K}_c(X_i-)\widehat{K}_c(X_j-)},
\end{equation}
\begin{equation}\label{L2C1}
    \widehat{M}^\mathcal{I}_{c,a}\!=\!\int\limits_0^\infty (\widehat{U}^\mathcal{I}_c(t))^2e^{-at}dt=\int\limits_0^\infty\bigg(\frac{1}{\binom{n}{2}}\sum_{i<j}\frac{h^\mathcal{I}(X_i,X_j;t)\delta_i\delta_j}{\widehat{K}_c(X_i-)\widehat{K}_c(X_j-)}\bigg)^2e^{-at}dt,
\end{equation} 
where $\widehat{K}_c(x-)=\prod\limits_{X_i< x}\bigg(1-\frac{1-\delta_i}{\sum_{k=1}^nI\{X_k\geq X_i\}}\bigg)$.  Naturally, the large values of $|\widehat{J}^{\mathcal{I}}_{n,a}|$ and $\widehat{M}^{\mathcal{I}}_{n,a}$  indicate that the null hypothesis should be rejected.
\begin{remark}
 An alternative approach for the construction of test statistics based on the mention characterizations is to estimate underlying Laplace transforms with their U-empirical Kaplan-Meier estimators (see e.g. \cite{datta2010inverse}). However, this estimator will coincide with the IPCW estimator when the largest observation is not censored. Therefore, the asymptotic and small sample properties of such statistics will be quite similar to those of the proposed tests.   
\end{remark}
Hereafter we present limiting distributions of proposed statistics under the null hypothesis of exponentiality. For simplicity, in the rest of this section, the test statistics and kernels will be labelled without $\mathcal{I}$.

First, we introduce several counting processes. Let $N^c_i(t)\!=\!I(X_i\!\leq\! t,\delta_i\!=\!0)$ be the right-censoring  counting process  for the $i$th individual and $Y_i(t)=I(X_i\geq t)$ is appropriate "at-risk" process. Then
$M_i^c(t)=N^c_i(t)-\int_0^tY_id\Lambda_c(u)$ is the associated martingale defined with respect to the filtration $$\mathcal{F}_t=\sigma(N^c_i(s), Y_i(s), 0\leq s\leq t, i=1,...,n),$$ where $\Lambda_c$ is the cumulative censoring hazard rate. Then, $N^c(t)=\sum_{i=1}^nN^c_i(t)$ is the number of censored observations in the interval $[0,t]$ and $Y(t)=\sum_{i=1}^nY_i(t)$ is the
number of objects at risk just prior to $t$. 

The process $\{\sqrt{n}\widehat{U}_c(t)\}$ is a random element in the Fr\'echet space $C(\mathbb{R}^+)$ of continuous functions on $\mathbb{R}^+$, endowed with the metric 
\begin{equation*}
    d(x,y)=\sum\limits_{n=1}^\infty\frac{1}{2^n}\min\{1,\max\limits_{0\leq t\leq n}|x(t)-y(t)|\}.
\end{equation*} 
In the next theorem we present its limiting behaviour. 
\begin{theorem}\label{distribution}
Let $X'_1,...,X'_n$  be the sample of i.i.d. random variables with distribution function $F(x)=1-e^{-\frac{x}{\mu} }$, $x>0$, and $(X_1,\delta_1),....(X_n,\delta_n)$ the corresponding censored sample. Suppose that 
\begin{enumerate}
    \item[a)]
    $\int\limits_0^\infty \frac{1}{K_c(u-)}dF(u)<\infty,$
    \item[b)] $\int\limits_{0}^\infty\frac{1-F(u)}{K_c^2(u)}dG(u)<\infty$, 
    \item[c)]  $\int\limits_0^\infty \frac{u^2}{K_c(u-)}dF(u)<\infty,$
    \item[d)] $\int\limits_{0}^\infty\frac{u^2(1-F(u))}{K_c^2(u)}dG(u)<\infty$.
\end{enumerate}
Then,
$\{\sqrt{n}\widehat{U}_{c}(t)\}$ converges in $C(\mathbb{R}^+\!)$ to some centered Gaussian process $\{\eta(t)\}$ with covariance equal to
\begin{equation}\label{kovarijacija}
    cov(\eta(t_1),\eta(t_2))=4E(\zeta(t_1)\zeta(t_2))
    ,\; t_1,t_2\in \mathbb{R}^+,
\end{equation}
where ${\zeta}(t_1)\!=\!\frac{h_1(X_1;t_1)\delta_1}{{K}_c(X_1-)}\!+\!\int_0^{\infty}\!{\omega}(u;t_1)d{M}^c_1(u)$ and $h_1\!(x;t)\!=\!E(h(X'_1,X'_2;t)\!|X'_1\!=\!x\!)$ is the first projection of kernel $h$ given in \eqref{jezgroh},  and for  $u\geq 0$
   $$ \omega(u;t)=\frac{1}{P(X_1\geq u)}\int\limits_u^\infty h_1(x;t)dF(x).$$
\end{theorem}
Covariance $cov(\eta(t_1),\eta(t_2))$ might be estimated with
\begin{align*}
    &\widehat{cov}(\eta(t_1),\eta(t_2))=\frac{4}{n}\sum\limits_{i=1}^n\widehat{\zeta}_i(t_1)\widehat{\zeta}_i(t_2)=\!\frac{4}{n}\!\sum\limits_{i=1}^n\!\Bigg(\!\frac{h_1(X_i;t_1)\delta_i}{\widehat{K}_c(X_i-)}\!+\!\int\limits_0^{\infty}\!\widehat{\omega}(u;t_1)d\widehat{M}^c_i(u)\!\Bigg)\!\Bigg(\!\frac{h_1(X_i;t_2)\delta_i}{\widehat{K}_c(X_i-)}\!+\!\int\limits_0^{\infty}\!\widehat{\omega}(u;t_2)d\widehat{M}^c_i(u)\!\Bigg)\!.
\end{align*}  
The suitable expression for $\int_0^{\infty}\widehat{\omega}(u;t)d\widehat{M}^c_i(u)$ is obtained from the following set of equalities
\begin{align*}
        \int\limits_0^{\infty}\widehat{\omega}(u;t_1)d\widehat{M}^c_i(u)&=\int\limits_0^{\infty}\widehat{\omega}(u;t_1)d{N}^c_i(u)-\int\limits_0^{\infty}\widehat{\omega}(u;t_1)Y_i(u)d\widehat{\Lambda}_c(u)
        \\&=\widehat{\omega}(X_i;t_1)(1-\delta_i)-\sum\limits_{j=1}^n\widehat{\omega}(X_j;t_1)I\{X_i\geq X_j\}\frac{1-\delta_j}{Y(X_j)},
\end{align*}
where $\widehat{\omega}(u;t_1)=
        \frac{1}{(1-F(u))\widehat{K}_c(u-)}\int_0^\infty h_1(x;t_1)I\{x>u\}dF(x)$
is consistent estimator of $\omega(u;t_1)$.
The last hold from $\widehat{K}_c(x)\!\stackrel{P}{\to}\!K_c(x)$ and $\max\limits_{i=1,...,n}\frac{K_c(X_i-)}{\widehat{K}_c(X_i-)}\!=\!O_p(1)$ (see \cite{zhou1991some}), and consequently
\begin{align}\label{postojanostKc}
    \frac{h_1(X_i;t)\delta_i}{\widehat{K}_c(X_i-)}=\frac{h_1(X_i;t)\delta_i}{{K}_c(X_i-)}+o_p(1).
\end{align}
In addition, since
$d\widehat{M}_i(u)=dN_i(u)+Y_i(u)d\widehat{\Lambda}(u)$
and $\widehat{\Lambda}(u)\stackrel{P}{\to}\Lambda(u)$ it holds $d\widehat{M}_i(u)\stackrel{P}{\to}dM_i(u)$. Hence, the consistency of $\widehat{cov}(\eta(t_1),\eta(t_2))$ is justified.\\

The limiting behaviour of \eqref{integralnaC1} and \eqref{L2C1} is given in the following theorems.
\begin{theorem}\label{raspodelaL2}
Under the conditions of Theorem \ref{distribution}, it follows that
\begin{equation*}
    n\int\limits_0^\infty \widehat{U}^2_c(t)e^{-at}dt\stackrel{D}{\to}\sum_{k=1}^\infty\upsilon_kW^2_k,
\end{equation*}
where $\upsilon_k, k=1,2,...,$ is the sequence of  eigenvalues of the integral operator $A$, defined for functions $q\in C(\mathbb{R}^+)$ for which $\int_{0}^{\infty}q(t)^2e^{-at}dt<\infty,$ by
\begin{align}\label{operator}
    Aq(t_1)=\int\limits_{0}^{\infty}cov(\eta(t_1),\eta(t_2))q(t_2)e^{-at_2}dt_2,
\end{align}
where $cov(\eta(t_1),\eta(t_2))$ given in \eqref{kovarijacija}, and $W_{k},\;k=1,2,...,$ are independent standard normal variables.
\end{theorem}
\begin{theorem}\label{data}
Suppose that conditions of the Theorem \ref{distribution} holds. Then $\sqrt{n}\widehat{J}_{c,a}\overset{D}{\to}\mathcal{N}(0,\sigma^2),$ where 
\begin{equation*}
\sigma^2=4Var\bigg(\frac{\Phi_1(X_1)\delta_1}{K_c(X_1-)}+\int\limits_0^{\infty}\frac{1}{P(X_1\geq u)}\int\limits_u^\infty \Phi_1(x)dF(x)dM^c_1(u)\bigg),
\end{equation*}
and $\Phi_1(x)=E(\Phi(X'_1,X'_2;a)|X'_1=x)$ is the first projection of kernel $\Phi$ given in \eqref{jezgroPhi}. 
\end{theorem}
The consistent estimator of $\sigma^2$ can be obtained following the procedure given after the Theorem \ref{distribution}.

\begin{remark}
The assumptions a)-d) of the Theorem 1 are not too restrictive. For example, in the case of Koziol-Green model which assume that $K_c(x)=(1-F(x))^\beta,$ where $\beta>0$ is unknown parameter, the assumptions a)-d) are satisfied if the censoring rate $p=\frac{\beta}{\beta+1}$ is less than 0.5. 
\end{remark}
The results provided in this section can be used for testing exponentiality in the large sample case. In the next section, we propose a suitable resampling procedure in small and moderate sample size case and show its consistency.

\begin{theorem}\label{postojanost}
Let $X$ and $Y$ be independent and identically distributed positive absolutely continuous random variables and $\mathcal{L}_{\omega_1}(t)=E(e^{-t\omega_1(X,Y)})$ and $\mathcal{L}_{\omega_2}(t)=E(e^{-t\omega_2(X,Y)})$. Then we have
\begin{align*}
  \widehat{J}_{c,a}&\overset{P}{\to} \int\limits_{0}^{\infty}(\mathcal{L}_{\omega_1}(t)-\mathcal{L}_{\omega_1}(t))e^{-at}dt,\\ 
  \widehat{M}_{c,a}&\overset{P}{\to} \int\limits_{0}^{\infty}(\mathcal{L}_{\omega_1}(t)-\mathcal{L}_{\omega_1}(t))^2e^{-at}dt.
\end{align*}
\end{theorem}
As a corollary we have  that the test $\widehat{M}_{c,a}$ will be consistent against all fixed alternatives  for which  the expression under the integral sign is different from zero, on the set of non-negligible measure. In particular case of Puri-Rubin characterization it means that the test is globally consistent, while in the case of Desu characterization it might not always hold  since the test is based on one value of $m$  that appears in the characterization (see the original version of characterization in \cite{desu1971}).
The test $\widehat{J}_{c,a}$ will be consistent against all alternatives where the theoretical counterpart
of $\widehat{J}_{c,a}$  is not equal to zero, which includes all distributions of practical interest.

\section{Empirical study}\label{sec: empirical}

Here, we study the performance of proposed tests via the extensive power comparison. We consider the case of testing simple hypothesis as well as the composite one, in the general case when there is no information about censoring distribution. One of the goals of this comparison study is to investigate the sensitivity of power performance  to different censoring rates. Therefore
in order to achieve certain  censoring rate we  use Koziol-Green model introduced in \cite{koziol1976cramer}.
We consider following competitor test statistics:
\begin{itemize}
\item The Cramer-von Mises test proposed in \cite{koziol1976cramer} with test statistic
\begin{align}\label{cvm}
    \omega^2=\int\limits_0^\infty(\widetilde{F}_n(t)-(1-e^{-\frac{t}{\mu}}))^2e^{-\frac{t}{\mu}}dt,
\end{align}
where $\widetilde{F}_n$ is Kaplan-Meier estimator modified with $\widetilde{F}_n(t)=1,$ for $t\geq X_{(n)}$ if the largest observation is censored. 

\item The $\chi^2$ test proposed in \cite{akritas1988pearson}. In the case of simple hypothesis the test statistics is
\begin{align}\label{akritas}
A_{nr}=\sum\limits_{j=1}^r\frac{(N_{1j}-n\widehat{p}_{1j})^2}{n\widehat{p}_{1j}},
\end{align}
where $N_{1j}=\sum_{i=1}^nI\{X_i\in A_j,\delta_j=1\}$, $\widehat{p}_{1j}=\frac{1}{\mu}\int_{A_j}(1-\widehat{H}(x))dx$ and $\widehat{H}(x)=\frac{1}{n}\sum_{i=1}^nI\{X_i\leq x\}$ is edf. Domain of exponential distribution is divided into intervals $A_j, \;j=1,..,r,$ in such way that each of them has the same probability under $H_0.$
In the case of composite hypothesis, the estimators  of the cell probabilities, are $\widetilde{p}_{1j}=\frac{1}{\widehat{\mu}}\int_{A_j}(1-\widehat{H}(x))dx
,$ where $\widehat{\mu}$ is the Kaplan-Meier estimator of $\mu$. 
Intervals $A_j,\;j=1,..,r,$ are formed similarly as in the case of simple hypothesis. The test statistic in this case is 
\begin{align}\label{akritaSlozena}
A_{nr}=\widetilde{V}'_nA\widetilde{V}_n,
\end{align}
where $A$ is generalised inverse of the matrix $\widehat{\Sigma}-\widehat{B}\widehat{I}_{\widehat{\theta}}\widehat{B}'$ with \\ $\widehat{\Sigma}=diag(\widetilde{p}_{11},...,\widetilde{p}_{1r})$ and $\widehat{B}$ is vector with column elements $\widehat{b}_j=\int_{A_j}(1-\widehat{H}(x))dx$, and vector $\widetilde{V}_n=\frac{1}{\sqrt{n}}((N_{11}-n\widetilde{p}_{11}),...,(N_{1r}-n\widetilde{p}_{1r})).$

\item The test based on maximal correlations proposed in \cite{strzalkowska2017goodness} with test statistic 
\begin{align}\label{graneStat}
   {Q^S_n}\!=\!\frac{\sqrt{n}Q_n}{\sqrt{\!\sigma^2_n\!}}\!=\!\frac{\sqrt{n}}{\sqrt{\sigma^2_n}}\!\sum\limits_{i\neq j}\!\omega_{in}\omega_{jn}\!((6Y_i\!-\!2)I\{Y_j\leq Y_i\}\!-\!6Y_iI\{Y_j\!>\!Y_i\}\!),
\end{align}
where 
$Y_i=1-e^{-\frac{X_i}{\mu}}$, and $\omega_{in}=F_n(Y_i)-F_n(Y_i-),$ and $F_n$ is Kaplan-Meier estimator of distribution $F$,
and $\sigma^2_n$ is consistent estimator of the variance of statistic $Q_n$ (see \cite{strzalkowska2017goodness}). 
\item The test based on properties of  DMTTF class of  life distributions  proposed in \cite{kattumannil2019simple} with test statistic
\begin{align}\label{katumanilStat}
\Delta_n=\frac{1}{\binom{n}{2}}\sum\limits_{i<j}\frac{\delta_i\delta_j}{\widehat{K}_c(X_i)\widehat{K}_c(X_j)}\Big(2\min\{X_i,X_j\}-\frac{1}{2}(X_i+X_j)\Big).
\end{align}
\end{itemize}
It should be  noticed that the test statistics \eqref{katumanilStat} is also constructed using IPCW approach, while the other three tests are constructed using usual Kaplan-Meier approach. Also, what is interesting to note is that test \eqref{katumanilStat} might be viewed as a test based on Desu characterization constructed via the moment-based approach. Moreover it can be obtained from $M^{\mathcal{D}}_{c,a}$ and $M^{\mathcal{D}}_{c,a}$ when $a\to\infty.$

The test's power performance is examined against  a Weibull ($W(\theta)$), a Gamma ($\Gamma(\theta)$), a Half-normal ($HN$), a Chen ($CH(\theta)$), a Linear failure rate ($LF(\theta)$), a Modified extreme value ($EV(\theta)$), a  Log-normal ($LN$) and a  Dhillon ($DL(\theta)$) alternatives whose densities can be found in e.g. \cite{cuparic2018new}.
Notice that this set of alternatives, commonly used in complete data case, reflects different discrepancies from an exponential distribution. That gives us the credentials to get some general conclusions about the performance of considered tests. 

Having in mind that the null  distributions  of test statistics depend on censoring distribution,  in small and moderate sample sizes, the usage of resampling procedure is necessary. In what follows we  adapt the procedure proposed in \cite{efron1981censored} to our  null-distribution settings.
\subsection{Testing simple hypothesis of exponentiality} 
In this section we consider the case of testing hypothesis $H_0$ that the sample comes from exponential $\mathcal{E}(1)$ distribution.
 Empirical  powers are obtained using the following bootstrap procedure:
\begin{enumerate}
    \item Based on $(X_1,\delta_1),...,(X_n,\delta_n)$ compute the test statistic $T_n$ ;
    \item Estimate critical value $q^*_{n,1-\alpha}$: 
    \begin{enumerate}
       \item generate sample $C^*_1,...,C^*_n$ from Kaplan-Meier  distribution function $G_n$;
        \item generate new sample $X'_1,...,X'_n$ from $\mathcal{E}(1)$  distribution;
        \item using (a) and (b) construct bootstrap sample $(X^*_1,\delta^*_1),...,(X^*_n,\delta^*_n)$, where $X^*_i=\min({X}^{'}_i,C^*_i)$, and $\delta^*_i=I\{{X}^{'}_i\leq C^*_i\}$; 
        \item  based on the sample from (c) determine the value of test statistic $$T^*_n=T_n((X^*_1,\delta^*_1),...,(X^*_n,\delta^*_n));$$
        \item  repeat steps (a)-(d)  B times;
        \item based on the obtained sequence of bootstrapped tests statistics estimate critical value $q^*_{n,1-\frac{\alpha}{2}}$;
    \end{enumerate}
    \item Reject $H_0$ if $T_n\geq q^*_{n,1-\alpha}$;
    \item Repeat steps 1-3 $N$ times. Estimated  test power  is percentage of  rejected  $H_0$.
\end{enumerate}
\begin{remark}
This procedure 
assumes  that the critical region is  of the form $W=\{T_n>q_{n,1-\alpha}\}$ for suitably chosen constant $q_{n,1-\alpha}$. However, in case of two sided-tests  $Q_{n}$ and $Q_n^S$  algorithm is appropriately modified. 
\end{remark}

Let  $\{\sqrt{n}\widehat{U}_c(t)\}$ be  the  process considered in  Theorem \ref{distribution} and $\{\eta(t)\}$ its weak limit.
The next theorem gives us the asymptotic behaviour of the bootstrapped process $\{\sqrt{n}\widehat{U}^*_c(t)\}$ which 
justifies the usage of the proposed bootstrap procedure.
\begin{theorem}\label{teoremaBoot}
Assume that the conditions from Theorem \ref{distribution} are fulfilled. Then, conditionally on the sample $(X_1,\delta_1),...,(X_n,\delta_n)$, we have that the $\{\sqrt{n}\widehat{U}^*_c(t)\}$ converges weakly to process $\{\eta(t)\}$ in $C(\mathbf{R}^+)$. \end{theorem}
 As a consequence, we have that the
null distributions of test statistics $\widehat{M}^{\mathcal{I}}_{c,a}$ and $\widehat{J}^{\mathcal{I}}_{c,a}$ can be approximated with proposed bootstrapped procedure.
The usage of this procedure for other considered test statistics can be justified analogously. Here we decided to keep up with the traditional application of \eqref{akritas} and therefore estimate its powers using asymptotic results.

In the  case of considered statistics with limiting normal distributions, for obtaining p-values, for larger sample sizes, 
one can also use their   standardized versions obtained using 
the  consistent estimator of asymptotic variance. This might significantly increase the computational performance of the testing procedure.

In our simulation study, 
we use described bootstrap approach,  with $B=1000$  and $N=1000$ replicates. It should be underlined that this procedure "keep" the censoring distribution which might have an impact on censoring rate of the bootstrapped sample obtained in step 2.(c).   The initial censoring level is controlled under assumed alternative distribution. 

The results, for the sample size $n=50$, the levels of censoring $\{0.1,0.2,0.3\}$ and the level of significance $\alpha=0.05$, are presented in Table \ref{tab: comparison50}. It can be noticed that, {for  all considered censoring rates almost all tests are well calibrated}. In particular, for  $p$ larger than 0.1 tests $Q_n$, $Q_n^S,$ $A_{n3}$ and $\Delta_n$ are the most liberal ones. However, for smaller $a$ $J^{\mathcal{I}}_{c,a}$ are better calibrated. Interestingly, for $M^{\mathcal{D}}_{c,a}$, and $M^{\mathcal{P}}_{c,1}$ the level of significance is kept even for $p=0.3$. 

As far the power performance is concerned, the general conclusion is that in most cases power decreases with the increase of initial censoring level, while in some cases the censoring level doesn't have significant impact on test power. Also, the ordering of tests shown to be not sensitive to the change of censoring rate.
As might be expected, no test outperforms all competitors for all selected alternatives. Although the test $A_{n3}$ is the optimal choice in the majority of cases,  new tests, for some choices of $a$, significantly outperform it in case of testing against Weibull and $LN(0.8)$ alternatives. 
It can also be noticed that the powers of new tests are usually not much affected by choice of $a$. However, some differences can be seen in the case of decreasing-increasing failure rate alternatives such as CH(0.5). Similar effect of tuning parameter can be noticed in the case of $\Gamma$(0.4) alternative. In those cases, the impact of the characterization, as well as the construction method, is notable. 
Having in mind the power study results and  the calibration of the tests, we believe that  $J^{\mathcal{D}}_{c,1},J^{\mathcal{P}}_{c,1},M^{\mathcal{P}}_{c,1}$, $M^{\mathcal{D}}_{c,1}$, and $A_{n3}$  deserve to be included in a battery of exponentiality tests for censored data. 
\begin{table}[]
\centering
	\scriptsize
	\caption{Percentage of rejected hypotheses for  $n=50$ for the simple hypothesis}
   \resizebox{1\textwidth}{!}{
		\begin{tabular}{cccccccccccccccccccc}
				\rotatebox[origin=c]{90}{p} &\rotatebox[origin=c]{90}{Alt.} & \rotatebox[origin=c]{90}{$Exp(1)$}&    \rotatebox[origin=c]{90}{$W(1.4)$}  & \rotatebox[origin=c]{90}{$\Gamma(2)$} & \rotatebox[origin=c]{90}{$HN$} & 
				\rotatebox[origin=c]{90}{$CH(0.5)$} & \rotatebox[origin=c]{90}{$CH(1)$} & \rotatebox[origin=c]{90}{$CH(1.5)$} & \rotatebox[origin=c]{90}{$LF(2)$} & \rotatebox[origin=c]{90}{$LF(4)$} &  \rotatebox[origin=c]{90}{$EV(1.5)$}& \rotatebox[origin=c]{90}{$LN(0.8)$} & \rotatebox[origin=c]{90}{$LN(1.5)$} & \rotatebox[origin=c]{90}{$DL(1)$} & \rotatebox[origin=c]{90}{$DL(1.5)$} & \rotatebox[origin=c]{90}{$W(0.8)$}  & \rotatebox[origin=c]{90}{$\Gamma(0.4)$}\\\hline

	\multirow{17}{*}{0.1} &	$\widehat{J}^{\mathcal{P}}_{c,1}$ & 5   & 74 &  96 & 39 & 16 & 9 & 97 & 45 & 58 & 75 & 85 & 67 & 76 & 100 & 37 & 80 \\ 
	&	$\widehat{J}^{\mathcal{P}}_{c,2}$ & 5 & 73 & 97 & 39 &  3 & 4 & 92 & 35 & 43 & 75 & 80 & 84 & 75 & 100 & 38 & 59\\ 
		&$\widehat{J}^{\mathcal{P}}_{c,5}$ &  4 & 70 & 97 & 35 &  1 & 2 & 79 & 23 & 25 & 71 & 72 & 93 & 75 & 100 & 40 & 37\\
		
		&$\widehat{J}^{\mathcal{D}}_{c,1}$ & 5 & 71 & 93 & 29& 52 &11 & 96 & 40 & 58 & 60 & 96 & 34 & 79 & 100 & 39 & 96\\
	    &$\widehat{J}^{\mathcal{D}}_{c,2}$ & 5& 72 & 97 & 33 & 18& 6 & 94 & 36 & 48 & 66 & 92 & 65 & 82 & 100 & 40 & 82\\
	    &$\widehat{J}^{\mathcal{D}}_{c,5}$ & 4 & 71 & 98 & 35 & 3 & 3 & 84 & 26 & 30 & 69 & 83 & 88 & 81 & 100 & 42 & 54\\
	   
	    &$\widehat{M}^{\mathcal{P}}_{c,1}$ & 5 & 74 & 96 & 39  & 26 & 11 & 97 & 47 & 61 & 75 & 89 & 65 & 76 & 100 & 36 & 85\\
	    &$\widehat{M}^{\mathcal{P}}_{c,2}$ & 4 & 74 & 97 & 39  & 6 & 6 & 95 & 40 & 49 & 75  & 83 & 80 & 75 & 100 & 38 & 68\\
	    &$\widehat{M}^{\mathcal{P}}_{c,5}$ & 4 & 71 & 98& 37  & 1 & 3 & 85 & 28 & 31 & 73 & 76 & 91 & 75 & 100 & 39 & 42\\
	    
	    &$\widehat{M}^{\mathcal{D}}_{c,1}$ & 5 & 69  & 92& 29 & 63 & 12 & 96 & 39 & 58 & 59 & 97 & 42 & 78 & 100 & 38 & 97\\
	   &$\widehat{M}^{\mathcal{D}}_{c,2}$ & 5&  71 & 97 & 32  & 28 & 8 & 95 & 37 & 53 & 65 & 95 & 67 & 82 & 100 & 40 & 88\\
	   &$\widehat{M}^{\mathcal{D}}_{c,5}$ & 4& 71 & 98 & 34  & 4 & 4 & 89 & 30 & 36 & 68 & 88 & 86 & 82 & 100 & 42 & 65\\
	   
	    &$\omega^2$ & 5 & 15 & 100 & 8  &100  & 100 & 100 & 98 & 100 & 15 & 95 & 70 & 100 & 100 & 13 & 100\\
	
	    &$A_{n3}$ & 6 & 53 & 100 &  55  &100 & 100& 100 & 100 & 100 & 92 & 81 & 98 & 98 & 100 & 29 & 100\\
	
	    &$Q_n$ & 5 & 65 & 96 & 42 &  99  & 100 & 100 & 97 & 100 & 79 & 81 & 42 & 68 & 99 & 35 & 51\\ 
	    &$Q^S_n$ & 5 & 68 & 88 & 47  & 99  & 100 & 100 & 99 & 100 & 84 & 87 & 39 & 63 &  97 & 28 & 33\\
	  
	    &$\Delta_n$ &  5 & 57 & 98 & 24  & 0  & 1 & 41 & 10 & 8 & 58 & 65 & 96 & 74 & 99 & 43 & 17\\
 	    \hline
	    
	    \multirow{17}{*}{0.2} & $\widehat{J}^{\mathcal{P}}_{c,1}$ & 5 & 69 & 96 & 34  & 6 & 10 & 93 & 39 & 51 & 71 & 85 & 58 & 71 & 99 & 31 & 61\\
		&$\widehat{J}^{\mathcal{P}}_{c,2}$ & 5 & 65 & 97 & 33 & 2 & 5 & 85 & 31 & 34 & 68 & 80 & 74 & 71 & 99 & 31 & 38\\
		&$\widehat{J}^{\mathcal{P}}_{c,5}$ & 5 & 60 & 97 & 28 & 1 & 5 & 67 & 20 & 21 & 62 & 72 & 85 & 70 & 99 & 32 & 20\\
		
		&$\widehat{J}^{\mathcal{D}}_{c,1}$ & 5 & 67 & 93 & 26 & 40 & 11 & 95 & 38 & 53 & 56 & 95 & 33 & 76 & 99 & 36 & 87\\
	    &$\widehat{J}^{\mathcal{D}}_{c,2}$ & 5 & 68 & 97 & 29 & 9 & 8 & 90 & 33 & 42 & 61 & 92 & 56 & 78 & 100 & 36 & 67\\
	    &$\widehat{J}^{\mathcal{D}}_{c,5}$ & 5 & 63 & 98 & 28 & 2 & 5 & 74 & 22 & 25 & 60 & 84 & 79 & 76 & 99 & 35 & 36\\
	    
	    &$\widehat{M}^{\mathcal{P}}_{c,1}$ & 5 & 69 & 96 & 34 & 12 & 12 & 94 & 41 & 55 & 71 & 88 & 56 & 70 & 99 & 31 & 69\\
	    &$\widehat{M}^{\mathcal{P}}_{c,2}$ & 5 & 67 & 97 & 33 & 3 & 7 & 90 & 34 & 42 & 70 & 83 & 70 & 71 & 99 & 31 & 48\\
	    &$\widehat{M}^{\mathcal{P}}_{c,5}$ & 5 & 61 & 97 & 30 & 1 & 5 & 74 & 23 & 26 & 64 & 75 & 82 & 71 & 99 & 32 & 26\\
	    
	    &$\widehat{M}^{\mathcal{D}}_{c,1}$ & 5 & 67 & 92 & 26 & 54 & 13 & 95 & 38 & 55 & 55 & 96 & 39 & 75 & 99 & 36 & 92\\
	    &$\widehat{M}^{\mathcal{D}}_{c,2}$ & 5 & 67 & 96 & 28 & 18 & 9 & 91 & 35 & 46 & 61 & 94 & 58 & 78 & 100 & 36 & 76\\
	    &$\widehat{M}^{\mathcal{D}}_{c,5}$ & 5 & 65 & 98 & 28 & 3 & 6 & 79 & 26 & 30 & 61 & 87 & 77 & 77 & 99 & 35 & 46\\
	    
	     &$\omega^2$ & 6 & 13 & 100 & 8 & 100 & 100 & 100 & 93 & 100 & 14 & 93 & 69 & 100 & 100 & 12 & 100\\
	   
	   &$A_{n3}$ & 7  & 42 & 100 & 47 & 100 & 100 & 100 & 98 & 100 & 77 & 79 & 96 & 97 & 100 & 21 & 100\\
	   
	    & $Q_n$ & 5 & 55 & 94 & 35 & 98 & 100 & 100 & 92 & 100 & 68 & 78 & 41 & 61 & 98 & 32 & 46\\
	    &$Q^S_n$ & 5 & 57 & 77 & 40 & 98 & 100 & 100 & 97 & 100 & 74 & 77 & 37 & 50 & 89 & 25 & 37\\
	 
	    &$\Delta_n$ & 4 & 47 & 98 & 20 & 0 & 4 & 33 & 11 & 14 & 47 & 65 & 89 & 73 & 98 & 32 & 9\\
	    \hline
	    
	    \multirow{17}{*}{0.3} & $\widehat{J}^{\mathcal{P}}_{c,1}$ & 5 & 62 &94 & 28  & 3 & 9 & 86 & 32& 45 & 59 & 80 & 46 & 71 & 99 & 25 & 37 \\
		&$\widehat{J}^{\mathcal{P}}_{c,2}$ & 5 & 59 & 95 & 25 & 2 & 6 & 74 & 26 & 32 & 55 & 75 & 62 & 70 & 99 & 25 & 19\\
		&$\widehat{J}^{\mathcal{P}}_{c,5}$ & 5 & 53 & 96 & 22 & 1 & 6 & 55 & 19 & 22 & 47 & 69 & 73 & 71 & 98 & 24 & 8\\
		
		&$\widehat{J}^{\mathcal{D}}_{c,1}$ & 5 & 65 & 93 & 27 & 26 & 12 & 91 & 35 & 51 & 49 & 94& 25 & 75 & 99 & 31 & 75\\
	    &$\widehat{J}^{\mathcal{D}}_{c,2}$ & 5 & 64 & 96 & 26 & 6 & 9 & 82 & 28 & 39 & 52 & 89 & 46 & 77 & 99 & 30 & 46\\
	    &$\widehat{J}^{\mathcal{D}}_{c,5}$ & 5 & 58 & 97 & 23 & 2 & 7 & 61 & 22 & 26 & 48 &80 & 66 & 76 & 99 & 27 & 17\\
	    
	    &$\widehat{M}^{\mathcal{P}}_{c,1}$ & 5 & 63 & 94 & 29 & 6 & 10 & 88 & 34 & 49 & 60 & 82& 45 & 71 & 99 & 25 & 46\\
	    &$\widehat{M}^{\mathcal{P}}_{c,2}$ & 5 & 59 & 95 & 26 & 2 & 7 & 80 & 29 & 37 & 57 & 77& 58 & 71 & 99 & 25 & 25\\
	    &$\widehat{M}^{\mathcal{P}}_{c,5}$ & 5 & 56 & 96 & 23 & 1 & 6 & 62 & 21 & 25 & 50 & 71& 70 & 71 & 99 & 24 & 10\\
	    
	    &$\widehat{M}^{\mathcal{D}}_{c,1}$ & 5 & 64 & 92 & 27 & 37 & 13 & 92 & 35 & 51 & 49 & 95& 30 & 74 & 99 & 32 & 82\\
	    &$\widehat{M}^{\mathcal{D}}_{c,2}$ & 5 & 64 & 96 & 26 & 11 & 10 & 84 & 30 & 42 & 52 & 92 & 47 & 77 & 99 & 30 & 57\\
	 	 &$\widehat{M}^{\mathcal{D}}_{c,5}$ & 5 & 60 & 97 & 24 & 3 & 7 & 68 & 24 & 29 & 49 & 84 & 65 & 77 & 99 & 29 & 25\\
	     
	      &$\omega^2$ & 4 & 13 & 100 & 7 & 100 & 100 & 98 & 82 & 98 & 11 & 92 & 67 & 100 & 100 & 11 & 100\\
	   
	   	      &$A_{n 3}$ & 7 & 38 & 100 & 38 & 100 & 100 & 99 & 91 & 100 & 59 & 80& 93 & 97 & 100 & 19 & 100\\
	   
	   & $Q_n$ & 5 & 48 & 92 & 28 & 96 & 98 & 100 & 82 & 97 & 50 & 66& 30 & 54 & 95 & 26 & 48\\
	    &$Q_n^S$ & 4 & 46 & 52 & 33 & 98 & 100 & 100 & 92 & 100 & 55 & 55& 29 & 32 & 68 & 17 & 49\\
	  
	    &$\Delta_n$ &4 & 41 & 97 & 16 & 1 & 6 & 33 & 13 & 18 & 32 & 65& 77 & 75 & 98 & 21 & 3\\
	    
	  \end{tabular}
	}
	\label{tab: comparison50}
\end{table}

\subsection{Testing composite hypothesis of exponentiality}
In this section we present results for testing $H_0:\;X\sim\mathcal{E}(\frac{1}{\mu})$ for arbitrary $\mu>0.$
It should be noted that most of the considered statistics can be made asymptotically scale-free, via application to the scaled sample.
However, since the distributions of test statistics also depend on the censoring distribution, which should be kept fixed in our simulation experiment, and we are dealing with samples of moderate size, this approach is not applicable.
Therefore, for the power estimation, we use the bootstrap algorithm analogous to the one used in simple hypothesis case, with the only difference in step 2. b.
Instead of it, we generate a new sample from $\mathcal{E}\big(\frac{1}{\widehat{\mu}}\big)$, where
$\widehat{\mu}=\frac{\sum_{i=1}^nX_i}{\sum_{i=1}^n\delta_i}$  
 is MLE of $\mu$ (see e.g. \cite{akritas1988pearson}).
The consistency of $\widehat{\mu}$ enables us to prove the consistency of bootstrap procedure using similar arguments like in the proof of the Theorem \ref{teoremaBoot}. Hence the proof is omitted. 

The results are presented in Table \ref{tab: comparison50slozena}.
Like in the case of testing simple hypothesis, most of considered tests are well calibrated. The only exception here is $\omega^2$ test whose size, for the censoring rate $p=0.3$, is slightly bellow the nominal level of significance. 

Given the fact that almost all tests are well-calibrated, we can recommend tests based on their power performance.
 The most powerful, with the exception of $LN(1.5)$ alternative, are proposed tests based on U-empirical Laplace transforms. In this particular case, $\omega^2$ and $\Delta_n$ are the best performing ones, while the power of proposed tests are the increasing functions of $a$.  Taking into account that, in non-censored case, $2\Delta_n^2$  is the limiting test statistic of $a^3M^{\mathcal{D}}_a$ when $a$ tends to infinity (the proof is analogous to the proof of the Theorem 4.1 from \cite{cuparic2018new}), we may expect that new tests and $\Delta_n$ will have the similar power performance against $LN(1.5)$ with larger value of $a$.
  
 In the case of the small censoring rate $p=0.1$ new tests based on Puri-Rubin characterization are in general more powerful than those based on Desu characterization. When $p=0.3$ tests based on Desu characterization take the lead position. 
In case that the practitioner doesn't have any preknowledge about the data, when $p=0.1$ we recommend $J^{\mathcal{P}}_{c,1},M^{\mathcal{P}}_{c,2}$,  for $p=0.2$ we recommend $J^{\mathcal{P}}_{c,1},M^{\mathcal{D}}_{c,1}$, while for larger censoring rates tests  $J^{\mathcal{D}}_{c,1}$ and $M^{\mathcal{D}}_{c,1}$ are the optimal choice.
\begin{table}[]
	\centering
	\scriptsize
	\caption{Percentage of rejected hypotheses for  $n=50$ for the composite hypothesis}
 \resizebox{1\textwidth}{!}{
		\begin{tabular}{ccccccccccccccccccc}
				\rotatebox[origin=c]{90}{p} &\rotatebox[origin=c]{90}{Alt.} & \rotatebox[origin=c]{90}{$Exp(1)$}&    \rotatebox[origin=c]{90}{$W(1.4)$}  & \rotatebox[origin=c]{90}{$\Gamma(2)$} & \rotatebox[origin=c]{90}{$HN$} & \rotatebox[origin=c]{90}{$CH(0.5)$} & \rotatebox[origin=c]{90}{$CH(1)$} & \rotatebox[origin=c]{90}{$CH(1.5)$} & \rotatebox[origin=c]{90}{$LF(2)$} & \rotatebox[origin=c]{90}{$LF(4)$} &  \rotatebox[origin=c]{90}{$EV(1.5)$}& \rotatebox[origin=c]{90}{$LN(0.8)$} & \rotatebox[origin=c]{90}{$LN(1.5)$} & \rotatebox[origin=c]{90}{$DL(1)$} & \rotatebox[origin=c]{90}{$DL(1.5)$} & \rotatebox[origin=c]{90}{$W(0.8)$}  & \rotatebox[origin=c]{90}{$\Gamma(0.4)$}\\\hline

	\multirow{17}{*}{0.1}& $\widehat{J}^{\mathcal{P}}_{c,1}$ & 5 & 77 & 94 & 45 & 91 & 30& 100 & 64 & 83 & 79 & 82 & 56 & 66 & 99 & 40 & 98\\ 
	&	$\widehat{J}^{\mathcal{P}}_{c,2}$ & 5 & 76 & 93 & 48 & 90 & 30 & 100 & 66 & 82 & 81 & 73 & 66 & 60 &99 & 40 & 97\\ 
		&$\widehat{J}^{\mathcal{P}}_{c,5}$ & 4 & 73 & 90 & 47 & 88 & 29 & 100 & 65 & 81 & 81 & 63 & 78 & 50 &98 & 40 & 97\\
		
		&$\widehat{J}^{\mathcal{D}}_{c,1}$ & 4 & 72 & 92 & 30 & 94 & 25 & 99 & 52 & 77 & 63 & 93 & 35 & 73 & 99& 41 & 100\\
	    &$\widehat{J}^{\mathcal{D}}_{c,2}$ & 5 & 75 & 94 & 36 & 92 & 27 & 100 & 59 & 82 & 72& 88 & 54 & 70 & 99& 42 & 99\\
	    &$\widehat{J}^{\mathcal{D}}_{c,5}$ & 4 & 74 & 93 & 44 & 90 & 27& 100 & 62 & 81 & 78 & 73 & 72 & 61 & 99& 43 & 98\\
	    
	    &$\widehat{M}^{\mathcal{P}}_{c,1}$ & 5 & 77 & 93 & 44 & 92 & 29 & 100 & 63 & 83 & 78& 86 & 56 & 68 & 99& 39 & 99\\
	    &$\widehat{M}^{\mathcal{P}}_{c,2}$ & 5 & 77 & 93 & 48 & 91 & 30 & 100 & 66 & 83 & 81 & 77& 64 & 63 & 99& 41 & 98\\
	    &$\widehat{M}^{\mathcal{P}}_{c,5}$ & 4 & 74 & 91 & 47 & 89 & 29 & 100 & 66 & 82& 82 & 66 & 74 & 53 & 98& 39 & 97\\
	    
	    &$\widehat{M}^{\mathcal{D}}_{c,1}$ & 4& 71 & 92 & 29 & 96 & 23 & 99 & 50 & 76 & 62 & 95 & 42 & 72 & 99& 40 & 100\\
	    &$\widehat{M}^{\mathcal{D}}_{c,2}$ & 5 & 74 & 93 & 35 & 93 & 26 & 100 & 58 & 81 & 71 & 90 & 54 & 70 & 99& 41 & 99\\
	    &$\widehat{M}^{\mathcal{D}}_{c,5}$ & 5 & 76 & 93 & 42 & 91 & 27 & 100 & 61 & 81 & 77 & 79& 69 & 64 & 99& 42 & 98\\
	    
	     &$\omega^2$ & 5 & 69 & 87 & 40 & 93 & 25 & 99 & 58 & 79 & 77 & 76 & 87 & 51 & 97& 39 & 98 \\
	    
	    &$A_{n3}$ & 5 & 61 & 74 & 41 & 83 & 25 & 97 & 53 & 73 & 74 & 53 & 82 & 39 & 88& 30 & 93\\ 
	     
	     &$Q_n$ & 5 & 68 & 91 & 26 & 89 & 15 & 97 & 44 & 67 & 58 & 85 &20  & 64 &98 & 32 & 95\\ 

	    &$\Delta_n$ & 4 & 66 & 81 & 42 & 87 & 25 & 99 & 60 & 79 & 79 & 43 & 87& 35 &91 & 40 & 96\\
	    \hline
	    
	    \multirow{17}{*}{0.2} & $\widehat{J}^{\mathcal{P}}_{c,1}$ & 5 & 69 & 89 & 39 & 84 & 26 & 99 & 52 & 73 & 73 & 80 & 44 & 64 & 97 & 32 & 92\\
		&$\widehat{J}^{\mathcal{P}}_{c,2}$ & 5 & 66 & 88 & 39 & 82 & 25 & 98 & 52 & 71 & 75 & 70 & 53 & 57 & 96 & 31 & 90\\
		&$\widehat{J}^{\mathcal{P}}_{c,5}$ & 5 & 62 & 83 & 35 & 80 & 23 & 96 & 49 & 67 & 73 & 57 & 64 & 49 & 93 & 30 & 88\\
		
		&$\widehat{J}^{\mathcal{D}}_{c,1}$ & 6 & 68 & 90 & 30 & 90 & 22 & 98 & 46 & 67 & 57 & 94 & 28 & 74 & 98 & 36 & 96\\
	    &$\widehat{J}^{\mathcal{D}}_{c,2}$ & 6 & 69 & 91 & 35 & 87 & 23 & 98 & 49 & 69 & 66 &87 & 43 & 70 & 98 & 36 & 94\\
	    &$\widehat{J}^{\mathcal{D}}_{c,5}$ & 5 & 64 & 89 & 36 & 82 & 22 & 97 & 48 & 67 & 70 & 71 & 60 & 59 & 97 & 33 & 91\\
	    
	    &$\widehat{M}^{\mathcal{P}}_{c,1}$ & 5 & 69 & 89 & 38 & 86 & 26 & 99 & 52 & 73 & 72 & 82 & 44 & 64 & 97 & 32 & 93\\
	    &$\widehat{M}^{\mathcal{P}}_{c,2}$ & 5 & 67 & 88 & 39 & 83 & 26 & 98 & 53 & 72 & 74 & 73 & 51 & 59 & 96 & 32 & 91\\
	    &$\widehat{M}^{\mathcal{P}}_{c,5}$ & 5 & 63 & 84 & 37 & 81 & 24 & 97 & 50 & 69 & 74 & 61 & 62 & 51 & 94 & 30 & 89\\
	    
	    &$\widehat{M}^{\mathcal{D}}_{c,1}$ & 6 & 67 & 89 & 29 & 91 & 21 & 98 & 45 & 66 & 56 & 95 & 33 & 73 & 98 & 37 & 97\\
	    &$\widehat{M}^{\mathcal{D}}_{c,2}$ & 6 & 68 & 91 & 34 & 88 & 22 & 98 & 48 & 69 & 64 & 90 & 44 & 71 & 98 & 36 & 95\\
	    &$\widehat{M}^{\mathcal{D}}_{c,5}$ & 5 & 66 & 89 & 36 & 83 & 23 & 98 & 49 & 68 & 69 & 77 & 58 & 62 & 97 & 34 & 91\\
	    
	     &$\omega^2$ & 5 & 57 & 78 & 33 & 91 & 21 & 94 & 43 & 63 & 65 & 70 & 78 & 48 & 93 & 31 & 96\\
	     
	     &$A_{n3}$ & 6 & 51 & 68 & 32 & 79 & 23 & 95 & 44 & 63 & 67 & 49 & 71 & 38 & 85 & 23 & 89\\ 
	     
	      &$Q_n$ & 5 & 60 & 84 & 27 & 86 & 15 & 96 & 37 & 57 & 54 & 80 & 22 & 59 & 96 & 30 & 90\\ 
	   
	    &$\Delta_n$ & 4 & 52 & 68 & 29 & 78 & 19 & 93 & 43 & 59 & 65 & 38 & 71 & 34 & 81 & 30 & 86\\
	    \hline
	    
	    \multirow{17}{*}{0.3} & $\widehat{J}^{\mathcal{P}}_{c,1}$ & 6 & 61 & 85 & 33 & 71 & 22 & 96 & 43 & 60 & 64 & 70 & 35 & 58 & 96 & 27 & 85\\
		&$\widehat{J}^{\mathcal{P}}_{c,2}$ & 5 & 57 & 82 & 32 & 67 & 20 & 94 & 40 & 56 & 63 & 60 & 43  & 50 & 94 & 26 & 82\\
		&$\widehat{J}^{\mathcal{P}}_{c,5}$ & 5 & 52 & 74 & 28 & 64 & 17 & 89 & 36 & 50 & 58 & 50 &  52 & 42 & 87 & 24 & 79\\
		
		&$\widehat{J}^{\mathcal{D}}_{c,1}$ & 5 & 64 & 89 & 28 & 80 & 20 & 96 & 39 & 62 & 53 & 91 & 24 & 72 & 98 & 34 & 93\\
	    &$\widehat{J}^{\mathcal{D}}_{c,2}$ & 5 & 62 & 89 & 31 & 74 & 20 & 95 & 41 & 60 & 58 & 81 & 36 & 66 & 98 & 32 & 88\\
	    &$\widehat{J}^{\mathcal{D}}_{c,5}$ & 5 & 56 & 84 & 29 & 68 & 18 & 92 & 37 & 52 & 59 & 62 & 51 & 54 & 94 & 27 & 83\\
	    
	    &$\widehat{M}^{\mathcal{P}}_{c,1}$ & 5 & 61 & 85 & 33 & 73 & 23 & 96 & 44 & 61 & 64 & 73 & 35 & 59 & 96 & 26 & 86\\
	    &$\widehat{M}^{\mathcal{P}}_{c,2}$ & 5 & 59 & 83 & 33 & 69 & 21 & 95 & 41 & 58 & 63 & 63 &  41& 52 & 94 & 26 & 84\\
	    &$\widehat{M}^{\mathcal{P}}_{c,5}$ & 5 & 54 & 77 & 30 & 65 & 18 & 91 & 37 & 52 & 60 & 53 & 50 & 43 & 89 & 24 & 80\\
	    
	    &$\widehat{M}^{\mathcal{D}}_{c,1}$ & 5 & 62 & 88 & 27 & 81 & 20 & 96 & 39 & 62 & 52 & 92 & 28 & 71 &98 & 34 & 94\\
	    &$\widehat{M}^{\mathcal{D}}_{c,2}$ & 5 & 62 & 88 & 29 & 77 & 20 & 96 & 40 & 60 & 57 & 85 & 37 & 67 & 98& 32 & 90\\
	    &$\widehat{M}^{\mathcal{D}}_{c,5}$ & 5 & 58 & 85 & 29 & 70 & 18 & 93 & 38 & 54 & 59 & 68 & 48 & 57 & 95& 29 & 85\\
	    
	    &$\omega^2$ & 3 & 41 & 61 & 21 & 82 & 15 & 79 & 30 & 45 & 47 & 55 & 64& 37 & 79& 26 & 94\\
	     
	    &$A_{n3}$ & 6 & 48 & 64 & 29 & 73 & 20 & 89 & 38 & 52 & 60 & 48 & 59 & 35 & 81& 20 & 86\\ 
	     
	      &$Q_n$ & 4 & 52 & 77 & 22 & 78 & 15 & 94 & 32 & 51 & 48 & 69 & 23 & 51 & 94& 28 & 85\\ 
	      
	    &$\Delta_n$ & 4 & 40 & 55 & 21 & 63 & 14 & 81 & 30 & 43 & 46 & 33 & 56 & 28 &66 & 22 & 77\\
		  \end{tabular}
	}
	\label{tab: comparison50slozena}
\end{table}
\section*{Appendix -- Proofs}\label{sec:apendix}
\begin{proof}[Proof of Theorem \ref{distribution}]
We show the statement of the theorem for the test based on Puri-Rubin characterization. In case of  Desu characterization, the proof is analogous.
The final aim is to show that $\sqrt{n}\widehat{U}_c(t)$ can be represented as 
\begin{align}\label{repMain}
\sqrt{n}\widehat{U}_c(t)=Z_n(t)+\sqrt{n}{R}_n(t),
\end{align} 
where $\{Z_n(t)\}$ converges to center Gaussian process and $\sqrt{n}{R}_n(t)=o_p(1).$
This will be done by: 
\begin{enumerate}
    \item \label{korak1} showing that $\sqrt{n}\widehat{U}_c(t)$ admits the presentation \eqref{repMain} where $Z_n(t)$ is a sum of i.i.d. random variables whose  summands don't depend on $\widehat{K}_c$;
    \item \label{korak2} showing that $\{Z_n(t)\}$ converges to centered Gaussian process for which is enough to show (see \cite{karatzas1991brownian} Theorem 4.15)  
    \begin{enumerate}
        \item that all finite-dimensional distributions of $Z_n(t)$ converges to multivariate normal distributions,
     \item $\{Z_n(t)\}$ is tight;
    \end{enumerate}
    \item applying Slutsky theorem for stochastic processes (see \cite{kosorok2008introduction}).
\end{enumerate}
\underline{The proof of \ref{korak1}:}
Due  to characterization, $\theta(t)=E(h(X'_1,X'_2;t))=0$.  Therefore, taking into account the mean preserving property of $\widehat{U}_{c}(t)$ (see e.g. \cite{datta2010inverse}) we have that $E(\widehat{U}_{c}(t))=0.$
Further, since $\widehat{U}_c(t)=U_c(t)+\widehat{U}_c(t)-U_c(t)$, we find appropriate representations for $\sqrt{n}U_c(s)$ as well as for
$\sqrt{n}(\widehat{U}_c(t)-U_c(t))$.
Let    
\begin{equation*}
\begin{split}
    \varphi_1(x_1,\delta_1)&\!=\!E\Big(\!\frac{{h}(X_1,X_2;t)\delta_1\delta_2}{K_c(X_1-)K_c(X_2-)}\Big|X_1=x_1,\delta_1\!\Big)
    =\frac{{h}_1(x_1;t)\delta_1}{K_c(x_1-)}.
    \end{split}
\end{equation*}
Then we have 
\begin{equation*}
\begin{split}
    \sqrt{n}U_c(t)\!-\!\frac{2}{\sqrt{n}}\sum_{i=1}^n \varphi_1(X_i,\delta_i;t)\!&=\!\sqrt{n}\frac{1}{\binom{n}{2}}\sum_{i<j}\bigg(\frac{{h}(X_i,X_j;t)\delta_1\delta_2}{K_c(X_i-)K_c(X_j-)}\!-\!\frac{{h}_1(X_i;t)\delta_1}{K_c(X_i-)}-\frac{{h}_1(X_j;t)\delta_2}{K_c(X_j-)}\bigg)\\&=\sqrt{n}\frac{1}{\binom{n}{2}}\sum_{i<j}\varphi_2(X_i,\delta_i;X_j,\delta_j;t).
    \end{split}
\end{equation*}
Since
\begin{align}\label{EH2}
    E(\varphi_2(X_1,\delta_1;X_2,\delta_2;t))^2&
    =E\bigg(\frac{{h}^2(X'_1,X'_2;t)}{K_c(X'_1-)K_c(X'_2-)}\bigg)-2E\bigg(\frac{{h}^2_1(X'_1;t)}{K_c(X'_1-)}\bigg)\nonumber
    \\&\leq 4\Big(E\Big(\frac{1}{K_c(X'_1-)}\Big)\Big)^2+8E\Big(\frac{1}{K_c(X'_1-)}\Big),
\end{align}
 and the last two summands are finite (due to assumptions of the theorem), it holds
 \begin{align}\label{Ucreprezentacija}
     \sqrt{n}U_c(t)=\frac{2}{\sqrt{n}}\sum_{i=1}^n \varphi_1(X_i,\delta_i;t)+\sqrt{n}R'_n(t),
 \end{align}
 where $\sqrt{n}R'_n(t)=o_p(1).$
Next, we have
\begin{align}\label{UcOminusUc}
    \sqrt{n}\widehat{U}_c(t)-\sqrt{n}U_c(t)
    &=\frac{\sqrt{n}}{\binom{n}{2}}\sum_{i<j}{h}(X_i,X_j;t)\delta_i\delta_j\bigg(\frac{K_c(X_i-)-\widehat{K}_c(X_i-)}{\widehat{K}_c(X_i-)K_c(X_i-)K_c(X_j-)}+\frac{K_c(X_j-)-\widehat{K}_c(X_j-)}{\widehat{K}_c(X_i-)\widehat{K}_c(X_j-)K_c(X_j-)}\bigg).
\end{align}
In the  expression above in the denominators,  we replace $\widehat{K}_c$ with its limit in probability $K_c$. The error of this approximation is equal to
\begin{align*}
    \sqrt{n}R''_n(t)&=
    -\frac{\sqrt{n}}{\binom{n}{2}}\sum_{i<j}{h}(X_i,X_j;t)\delta_i\delta_j\bigg(\frac{(\widehat{K}_c(X_i-)\!-\!K_c(X_i-)\!)\!(\widehat{K}_c(X_j-)\!-\!K_c(X_j-)\!)}{\widehat{K}_c(X_i-){K}_c(X_i-)K^2_c(X_j-)}\!+\!\frac{(\widehat{K}_c(X_i-)-K_c(X_i-))^2}{\widehat{K}_c(X_i-)K^2_c(X_i-)K_c(X_j-)}\\&\!+\!\frac{(\widehat{K}_c(X_j-)\!-\!K_c(X_j-)\!)^2}{\widehat{K}_c(X_i-)\widehat{K}_c(X_j-)K^2_c(X_j-)} \!\bigg)\!. 
\end{align*}
In order to show that $\sqrt{n}R''_n(t)$ uniformly converges to zero, it suffices to show that  process $\sqrt{n}\widetilde{R}''_n(t)$, defined with, 
\begin{equation*}
\begin{split}
\sqrt{n}\widetilde{R}''_n(t) =   \int B_n(x_1,x_2)\frac{{h}(x_1,x_2;t)}{\widehat{K}_c(x_1-)\widehat{K}_c(x_2-)}dW_n(x_1)dW_n(x_2)
\end{split}
\end{equation*}
where $W_n$ is empirical sub-distribution of the non-censored observations $X_i$ (i.e. those with $\delta_i=1$), and 
\begin{align*}
    B_n(x_1,x_2)&=\sqrt{n}\bigg(\!\frac{-(\widehat{K}_c(x_1-)-K_c(x_1-))(\widehat{K}_c(x_2-)\!-\!K_c(x_2-)\!)}{\widehat{K}_c(x_1-){K}_c(x_1-)K^2_c(x_2-)}+\!\frac{-(\widehat{K}_c(x_1-)\!-\!K_c(x_1-)\!)^2}{\widehat{K}_c(x_1-)K^2_c(x_1-)K_c(x_2-)}\\&+\!\frac{-(\widehat{K}_c(x_2-)\!-\!K_c(x_2-)\!)^2}{\widehat{K}_c(x_1-)\widehat{K}_c(x_2-)K^2_c(x_2-)}\!\bigg)\!\widehat{K}_c(x_1-)\widehat{K}_c(x_2-),
\end{align*}
 uniformly converges to zero.
That is because $\sqrt{n}\widetilde{R}''_n(t)$  is equal to
\begin{align*}
    \sqrt{n}\widetilde{R}''_n(t)&=
    -\frac{\sqrt{n}}{n^2}\sum_{i,j}{h}(X_i,X_j;t)\delta_i\delta_j\bigg( \frac{(\widehat{K}_c(X_i-)\!-\!K_c(X_i-)\!)(\widehat{K}_c(X_j-)\!-\!K_c(X_j-)\!)}{\widehat{K}_c(X_i-){K}_c(X_i-)K^2_c(X_j-)}\!+\!\frac{(\widehat{K}_c(X_i-)-K_c(X_i-))^2}{\widehat{K}_c(X_i-)K^2_c(X_i-)K_c(X_j-)}\\&\!+\!\frac{(\widehat{K}_c(X_j-)\!-\!K_c(X_j-)\!)^2}{\widehat{K}_c(X_i-)\widehat{K}_c(X_j-)K^2_c(X_j-)}\!\bigg)
\end{align*}
which  can be further expressed as
\begin{align*}
    {\sqrt{n}}\widetilde{R}''_n(t)\!=\!\frac{n(n\!-\!1)}{n^2} {\sqrt{n}}{R}''_n(t)\!+\!\frac{\sqrt{n}}{n^2}\sum_{i}h(X_i,X_i;t)\delta_i\frac{3(\widehat{K}_c(X_i-)\!-\!K_c(X_i-))^2}{\widehat{K}_c(X_i-)K^3_c(X_i-)}.
\end{align*}
Thus the term on the left hand side and the first term on the right hand side attain the same limit.

The process $B_n$ is tendentiously formed in such way that the application of Cauchy-Schwarz inequality will be enough to show that $\sqrt{n}R''_n(t) $ uniformly converges to zero.
After multiplying out in $B_n(x_1,x_2)$, the resulting first
summand might be written as 
\begin{align}\label{suma}
    -\sqrt{n}\frac{\widehat{K}_c(x_1-)\!-\!K_c(x_1-)}{\widehat{K}_c(x_1-)}
    \!\cdot\! \frac{\widehat{K}_c(x_2-)\!-\!K_c(x_2-)}{\widehat{K}_c(x_2-)}\!\cdot\! \frac{\widehat{K}_c(x_1-)}{{K}_c(x_1-)}\!\cdot\! \frac{\widehat{K}^2_c(x_2-)}{{K}^2_c(x_2-)}.
\end{align}
From the Theorem 2 (see \cite{ying1989note}) (with putting $K_c(x)=1-G(x)$) we have that 
$\sqrt{n}\frac{1-K(x)}{\widehat{K}_c(x)}(\widehat{K}_c(x)-K_c(x))$, where $ C(t)=\int_0^t\frac{d\Lambda(s)}{1-H(s-)}$ and $ K(t)=\frac{C(t)}{1+C(t)}$,
converges
to  Gaussian process with certain covariance matrix. That is equivalent to convergence of  $\sqrt{n}\frac{(\widehat{K}_c(x)-K_c(x))}{\widehat{K}_c(x)}$ (the first factor in \eqref{suma})  to some other   Gaussian process (see \cite{gill1983large}).
Therefore the second factor uniformly converges to 0, while the last two factors  converge to 1 (due to consistency of Kaplan-Meier estimator).
The similar reasoning hold for other two summands in the expression for $B_n$.

Applying Cauchy-Schwarz inequality we have 
\begin{align*}
|\sqrt{n}R''_n(t)|&\leq\bigg(\int B^2_n(x_1,x_2)dW_n(x_1)dW_n(x_2)\bigg)^{\frac{1}{2}}\bigg(\int \frac{{h}^2(x_1,x_2;t)}{\widehat{K}^2_c(x_1-)\widehat{K}^2_c(x_2-)}dW_n(x_1)dW_n(x_2)\bigg)^{\frac{1}{2}}
   \\&\leq\sup_{x_1,x_2\geq0}|B^2_n(x_1,x_2)|\bigg(\frac{1}{n^2}\sum_{i,j=1}^n\frac{{h}^2(X_i,X_j;t)\delta_i\delta_j}{\widehat{K}^2_c(X_i)\widehat{K}^2_c(X_j)}\bigg)^{\frac{1}{2}}.
\end{align*}
Since we have shown that  process $\{B_n(x_1,x_2)\}$ converges weakly  to zero process 
and 
$\frac{1}{n^2}\!\sum\limits_{i,j=1}^n\!\frac{{h}^2(X_i,X_j;t)\delta_i\delta_j}{\widehat{K}^2_c(X_i)\widehat{K}^2_c(X_j)}$ converges in probability to  $E\left(\frac{{h}^2(X'_i,X'_j;t)}{{K}_c(X'_i){K}_c(X_j)}\right)<4\Big(E\Big(\frac{1}{K_c(X'_1-)}\Big)\Big)^2$, 
we get that $\sqrt{n}R''_n(t)$ converges uniformly to 0. 

Further we have that for every $u>0$ 
\begin{align*}
    \sqrt{n}(\widehat{K}_n(u-)-K_n(u-))
    =\sqrt{n}K_c(u-)(\Lambda_c(u-)-\widehat{\Lambda}_c(u-))+o_p(1),
\end{align*}
where $\Lambda_c$ is cumulative hazard function ($\Lambda_c(u)=-\ln{K_c(u)}$), and $\widehat{\Lambda}_c$ Nelson-Aaalen estimate of hazard function. 
Then, \eqref{UcOminusUc} can be express as
\begin{equation}\label{U}
\begin{split}
    \sqrt{n}\widehat{U}_c(t)-\sqrt{n}U_c(t)&
    =\frac{\sqrt{n}}{\binom{n}{2}}\sum_{i<j}{h}(X_i,X_j;t)\delta_i\delta_j\bigg(\frac{\widehat{\Lambda}_c(X_i-)-\Lambda_c(X_i-)}{K_c(X_i-)K_c(X_j-)}+\frac{\widehat{\Lambda}_c(X_j-)-\Lambda_c(X_j-)}{{K}_c(X_i-)K_c(X_j-)}\bigg)+\sqrt{n}R'''_n(t)
\end{split}
\end{equation}
where $\sqrt{n}R'''_n(t)$ converges uniformly to zero.

The second term of \eqref{U} can be approximated with 
\begin{equation}\label{lambda_ocen-lambda}
\frac{2}{\sqrt{n}}\sum_{i=1}^n\frac{{h}_1(X_i;t)\delta_i}{K_c(X_i-)}(\widehat{\Lambda}_c(X_i-)-\Lambda(X_i-)),
\end{equation}
while for the error of this approximation it holds
\begin{equation}\label{razlika}
\begin{split}
    &\Big|\frac{2}{\sqrt{n}}\sum_{i=1}^n\frac{{h}_1(X_i;t)\delta_i}{K_c(X_i-)}(\widehat{\Lambda}_c(X_i-)-\Lambda_c(X_i-))-\frac{\sqrt{n}}{\binom{n}{2}}\sum_{i<j}{h}(X_i,X_j;t)\delta_i\delta_j\bigg(\frac{\widehat{\Lambda}_c(X_i-)-\Lambda_c(X_i-)}{K_c(X_i-)K_c(X_j-)}+\frac{\widehat{\Lambda}_c(X_j-)-\Lambda_c(X_j-)}{{K}_c(X_i-)K_c(X_j-)}\bigg)\Big|
       \\&\!{\leq}2\sup_{s\geq0}\!\sqrt{n}|\widehat{\Lambda}_c(s)\!-\!\Lambda_c(s)|\bigg(\frac{1}{{n}}\!\sum_{i=1}^n\!\Big(\!\frac{{h}_1(X_i;t)\delta_i}{K_c(X_i-)}\!-\!\frac{1}{n\!-\!1}\!\sum_{j\neq i}\! \frac{{h}(X_i,X_j;t)\delta_i\delta_j}{K_c(X_i-)K_c(X_j-)}\!\Big)^2\!\bigg)^\frac{1}{2}\!.
\end{split}
\end{equation}
Further it follows 
\begin{equation*}
\begin{split}
    E\bigg(\frac{1}{{n}}&\sum_{i=1}^n\Big(\frac{{h}_1(X_i;t)\delta_i}{K_c(X_i-)}-\frac{1}{n-1}\sum_{j\neq i} \frac{{h}(X_i,X_j;t)\delta_i\delta_j}{K_c(X_i-)K_c(X_j-)}\Big)^2\bigg)
    \\&=E\bigg[E\bigg(\bigg(\frac{{h}_1(X_1;t)\delta_1}{K_c(X_1-)}-\frac{1}{n-1}\sum_{j\neq 1} \frac{{h}(X_1,X_j;t)\delta_1\delta_j}{K_c(X_1-)K_c(X_j-)}\bigg)^2\bigg|X_1,\delta_1\bigg)\bigg].
\end{split}
\end{equation*}
Since for fixed $t$, $\frac{1}{n-1}\sum_{j\neq 1} \frac{{h}(x_1,X_j;t)\delta_1\delta_j}{K_c(x_1-)K_c(X_j-)}$ is $U-$statistic with mean value $\frac{{h}_1(x_1;t)\delta_1}{K_c(t_1-)}$, we have 
that the expression above is $O(n^{-1})$. 
Taking into account that $\sqrt{n}(\Lambda_c(u-)-\widehat{\Lambda}_c(u-))=O_p(1)$, we  conclude that \eqref{razlika} converges in probability  to 0.
Due to martingale representation of Nelson-Aalen estimation of hazard function,  \eqref{lambda_ocen-lambda} becomes
\begin{align}\label{finalnoLambda}
    \frac{2}{\sqrt{n}}\!\sum_{i=1}^n\!\frac{{h}_1(X_i;t)\delta_i}{K_c(X_i-)}\!\int\limits_0^{X_i-}\!\frac{d{M}^c(u)}{{Y}(u)}
  \! =\!\frac{2}{\sqrt{n}}\sum\limits_{i=1}^n\bigg(\int\limits_0^{\infty}\omega(u;t)dM^c_i(u)\bigg)\!+\!o_p(1)
\end{align}
where $y(u)=EY_1(u),$ 
$\omega(u;t)=\frac{1}{y(u)}\int_{u}^\infty{h}_1(x;t)dF(x)$.

Combining \eqref{Ucreprezentacija}, \eqref{U},  and \eqref{finalnoLambda}, we  get  \eqref{repMain} 
where 
\begin{align}\label{Zn}
    Z_n(t)=\frac{2}{\sqrt{n}}\sum\limits_{i=1}^n\bigg(\frac{{h}_1(X_i;t)\delta_i}{K_c(X_i-)}+\int\limits_0^{\infty}\omega(u;t)dM^c_i(u)\bigg).
\end{align}

\underline{The proof of \ref{korak2}:} 
Since, $P(X_1\geq u)=P(X'_1\geq u)P(C_1\geq u)=(1-F(u))K_c(u)$ and $\lambda_c(u)du=\frac{d(-\ln K_c(u-))}{du}=\frac{g(u)}{K_c(u-)}du$, we have
\begin{align*}
        &Var\left(\int\limits_0^{\infty}\omega(u;t)dM^c_1(u)\right)=\int\limits_0^\infty\omega^2(u;t)y(u)\lambda_c(u)du\\&\!=\!
         \int\limits_0^\infty\!\frac{g(u)}{(1\!-\!F(u))K^2_c(u-)}\!\bigg(\!\int\limits_u^\infty\!{h}_1(x;t)\!f(x)dx\!\bigg)^2\!du\!\leq\!\int\limits_0^\infty\!\frac{4(1\!-\!F(u))g(u)}{K^2_c(u-)}du< \infty.
\end{align*}
Hence, combining with \eqref{EH2} we have that
for fixed $t$ we have that $\frac{1}{\sqrt{n}}Z_n(t)$ is the sum of i.i.d. random variables  with finite variance, and, as such, has limiting normal distribution. 
Similarly, applying multivariate central limit theorem, one can show that the same hold for all  finite dimensional distributions of $Z_n$.

 To prove that the sequence $\{Z_n\}$ is tight, it is sufficient to show that for each $k\geq 1$ the sequence $Z_n$ restricted to $[0,k]$ is tight.
For the sake of showing it, we have to prove 
\begin{enumerate}
    \item[a)] $\sup\limits_{n\geq 1}E|Z_n(0)|^\nu=M<\infty$;
    \item[b)] $ \sup\limits_{n\geq 1}E|Z_n(t_2)-Z_n(t_1)|^\alpha\leq \psi(k)|t_2-t_1|^{1+\beta}, \forall k>0, 0\leq t_1,t_2\leq k,$
\end{enumerate}
for some positive constants $\alpha, \beta, \nu$ and function $\psi(k)$ (see \cite{karatzas1991brownian}).
Since $Z_n(0)=0,$ the first condition holds. 

Denoting $t_2=t+\rho$ and $t_1=t$,   $Z_n(t\!+\!\rho)\!-\!Z_n(t)$ can be expressed as 
\begin{equation}\label{Zn-Zn}
    \begin{split}
        \frac{2}{\sqrt{n}}\!\sum\limits_{i=1}^n\!\frac{(h_1(\!X_i;t\!+\!\rho\!)\!-\!h_1(X_i;t)\!)\delta_i}{K_c(X_i-\!)}\!+\!\frac{2}{\sqrt{n}}\!\sum\limits_{i=1}^n\!\int\limits_0^{\infty}\!(\omega(u;t\!+\!\rho)\!-\!\omega(u;t)\!)dM_i(u)\!.
    \end{split}
\end{equation}
We prove that both terms are tight. We can see that
\begin{equation*}
    \begin{split}
        E\bigg(\frac{2}{\sqrt{n}}\sum\limits_{i=1}^n\!\frac{(h_1(X_i;t+\rho)\!-\!h_1(X_i;t))\delta_i}{K_c(X_i-)}\bigg)^2   \!=\!4E\bigg(\!\frac{(h_1(X'_1;t+\rho)\!-\!h_1(X'_1;t))^2}{K_c(X^*_1-)}\!\bigg).
    \end{split}
\end{equation*}
Using mean value theorem it follows that
\begin{align*}
       |h_1(x;t+\rho)-h_1(x;t)|
      \leq |\rho|(36x+17),
\end{align*}
where $\xi\in(0,\rho)$. Then
\begin{equation*}
    \begin{split}
    &E\bigg(\frac{(h_1(X^*_1;t+\rho)-h_1(X^*_1;t))^2}{K_c(X^*_1-)}\bigg)
        \leq |\rho|^2E\bigg(\frac{(36X'_1+17)^2}{K_c(X'_1-)}\bigg).
    \end{split}
\end{equation*}
Thus
\begin{equation*}
    \begin{split}
        E\bigg(\frac{2}{\sqrt{n}}\sum\limits_{i=1}^n\frac{(h_1(X_i;t+\rho)-h_1(X_i;t))\delta_i}{K_c(X_i-)}\bigg)^2&\leq |\rho|^2\psi_1(k),
    \end{split}
\end{equation*}
where $\psi_1(k)=E\big(\frac{(36X'_1+17)^2}{K_c(X'_1-)}\big)<\infty$ 
from  the conditions of the theorem. 

For the second term of \eqref{Zn-Zn} it follows that
\begin{equation*}
    \begin{split}
        E\!\bigg(\!\frac{2}{\sqrt{n}}\!\sum\limits_{i=1}^n\!\int\limits_0^{\infty}\!(\omega(u;\!t\!+\!\rho)\!-\!\omega(u;\!t)\!)dM_i(u)\!\!\bigg)^2\!
        \!\!=\!4\!\int\limits_0^{\infty}\!(\omega(u;\!t\!+\!\rho)\!-\!\omega(u;\!t)\!)^2y(u)\!\lambda_c\!(u)du.
    \end{split}
\end{equation*}

Similar as before using the mean value theorem
\begin{align*}
    |\omega(u;t+\rho)-\omega(u;t)|
    \leq
    \frac{|\rho|}{P(X_1\geq u)}\int\limits_u^\infty (36x+17)dF(x)= \frac{|\rho|e^{-u}(36u+53)}{P(X_1\geq u)} .
\end{align*} 
Consequently, we have
\begin{align*}
        &E\bigg(\frac{2}{\sqrt{n}}\sum\limits_{i=1}^n\int\limits_0^{\infty}(\omega(u;t+\rho)-\omega(u;t))dM_i(u)\bigg)^2\leq 4\int\limits_0^{\infty}\bigg(\frac{|\rho|e^{-u}}{P(X_1\geq u)}(36u+53)\bigg)^2P(X_1\geq u)\frac{g(u)}{ K_c(u-)}du
        =|\rho|^2\psi_2(k),
\end{align*}
where $\psi_2(k)=\int_0^{\infty}\frac{e^{-2u}}{P(X_1\geq u)}(36u+53)^2\frac{g(u)}{ K_c(u-)}du<\infty$. 
Therefore
\begin{align*}
    E(Z_n(t\!+\!\rho)\!-\!Z_n(t))^2\leq 2|\rho|^2(\psi_1(k)+\psi_2(k))
\end{align*}
and condition b) for tightness holds. 
Finally, using Slutsky theorem for stochastic processes that appear in the  expression \eqref{repMain}  we complete the proof. 
\end{proof}
\begin{proof}[Proof of Theorem \ref{raspodelaL2}]
First step in the proof is to show
\begin{equation}\label{konver}
    \int\limits_0^\infty Z^2_n(t)e^{-at}dt\stackrel{D}{\to}\int\limits_0^\infty \eta^2(t)e^{-at}dt,
\end{equation}
where $\{\eta(t)\}$ is centered Gausian process with covariance \eqref{kovarijacija}. 

From  conditions a) and b) of Theorem \ref{distribution} we have that
\begin{align*}
   \int\limits_0^\infty cov(\eta(t),\eta(t))e^{-at}dt<\infty,
\end{align*}
hence, applying Tonellis theorem we conclude that
\begin{align*}
    \int\limits_0^\infty \eta^2(t)e^{-at}dt<\infty \text{ a.s.}
\end{align*}
The rest of this part  of the proof  goes along the same lines as in  \cite[Theorem 2.2.]{henze1997new}. 
Next, since $|\sqrt{n}\widehat{U}_c(t)-Z_n(t)|=\sqrt{n}R_n(t)$,
using triangular inequality we get that 
\begin{equation}\label{konverU}
    \int\limits_0^\infty (\sqrt{n}\widehat{U}_c(t))^2e^{-at}dt\stackrel{D}{\to}\int\limits_0^\infty \eta^2(t)e^{-at}dt.
\end{equation}
Covariance function is bounded, symmetric and positive definite, therefore from Mercer's theorem (\cite{van2013detection}) we have the following decomposition 
\begin{equation}\label{dekompozicija}
   cov(\eta(t_1),\eta(t_2))=\sum\limits_{i=1}^\infty{\upsilon_i}q_i(t_1)q_i(t_2),
\end{equation}
where $\upsilon_i$ are eigenvalues and $q_i(t)$ are corresponding eigenvectors of the integral operator $A$ defined in \eqref{operator}.
From the Karhunen-Loeve expansion of stochastic process (see e.g. \cite{van2013detection}) we have
\begin{equation}\label{razvojEta}
    \eta(t)=\sum\limits_{i=1}^\infty\sqrt{\upsilon_i}W_iq_i(t),
\end{equation}
where $W_i$ are independent random variables with $\mathcal{N}(0,1)$ distribution. Substituting the expressions \eqref{razvojEta} in \eqref{konverU} we complete the proof.  
\end{proof}
\begin{proof}[Proof of Theorem \ref{data}]

The proof is similar to the Step \ref{korak1} in the proof of Theorem \ref{distribution}  hence we omit it here.
\end{proof}

\begin{proof}[Proof of Theorem \ref{postojanost}]
 From the law of large numbers for U-statistics, and mean preserving property  we have that $$Z_n(t)\overset{P}{\to} E(h_1(X';t))=E(h(X'_1,X_2';t))=\mathcal{L}_{\omega_1}(t)-\mathcal{L}_{\omega_1}(t),$$
where $Z_n(t)$ is given by \eqref{Zn}.
In addition, from \eqref{postojanostKc} it follows that for each $t>0,$ $\widehat{U}_c(t)=Z_n(t)+o_p(1)$. Therefore, $\widehat{U}_c(t)$ converges in probability to $\theta(t)=E(h(X_1',X_2';t))=\frac{1}{2}E\Big(e^{-tX'_{1}}+e^{-tX'_{2}}-2e^{-t \psi^\mathcal{I}(X'_{1},X'_{2})}\Big) <2.$  Applying the continuous mapping  theorem  we complete the proof.
\end{proof}

\begin{proof}[Proof of Theorem \ref{teoremaBoot}]
Let $\{Z_n^*(t)\},$ be the bootstrapped version of process $\{\!Z_n\!(t)\!\}$ defined with 
\[Z^*_n(t)=\frac{2}{\sqrt{n}}\sum\limits_{i=1}^n\Bigg(\frac{h_1(X^*_i;t)\delta^*_i}{\widehat{K}_c(X^*_i-)}+\int\limits_0^\infty\widehat{\omega}(u;t)d{M}^{c*}_i(u)\Bigg),\] where $d{M}^{c*}_i(u)=d{N}^{c*}_i(u)+Y^{*}(u)d\widehat{\Lambda}_c(u)$ and ${M}^{c*}_i(u)$ is martingale with respect to appropriate filtration. 
We have to show that
\begin{enumerate}
\item\label{b1}  $\{Z_n^*(t)\},$ conditionally on $(X_1,\delta_1),...,(X_n,\delta_n),$ converges weakly to $\{\eta(t)\}$;
    \item\label{b2} $\sqrt{n}\widehat{U}^*_{n,c}(t)=Z^*_n(t)+\sqrt{n}R^*_n(t)$, where $||\sqrt{n}R^*_n(t)||\stackrel{P^*}{\rightarrow}0$.
\end{enumerate}

The proof of the first part consist of showing that all finite dimensional distributions are multivariate normal, that $\{Z^*_n(t)\}$ is tight and that its limiting covariance coincides with \eqref{kovarijacija}.

The first   follows from Lindberg-Feller CLT. In particular, since $\frac{\delta^*_i}{\widehat{K}_c(X^*_i)}=1-\int_0^\infty\frac{dM^{c*}_i(u)}{\widehat{K}_c(u)}$ (see \cite{robins1992recovery}), 
$Z_n^*(t)$ can be expressed as
\begin{align*}
    Z^*_n(t)\!&=\!
    \frac{2}{\sqrt{n}}\!\sum\limits_{i=1}^n\!\Bigg(\!h_1(X'_i;t)\!-\!\int\limits_0^\infty\!\Bigg(\!h_1(X'_i;t)\!-\!\frac{1}{1\!-\!F(u)}\!\int\limits_u^{\infty}\!h_1(v;t)\!dF(v)\!\Bigg)\!\frac{dM_i^{c*}(u)}{\widehat{K}_c(u-)}\!\Bigg)\!.
\end{align*}
From the martingale property of $M^{c*}$, 
we have  $E^*\!\left(\!Z_n^*\!(t)\!\right)\!=\!0$. Next
\begin{align}\label{LFCLT}
    &\lim\limits_{n\to\infty}\!\frac{1}{n\sigma_*^2(t)}\!\sum\limits_{i=1}^n\!E^*\!\Bigg[\!\!\Bigg(\!{h}_1(X'_i;t)\!-\!\int\limits_0^\infty\!\Bigg(\!{h}_1(X'_i;t)\!-\!\frac{\int\limits_u^{\infty}\!{h}_1(v;t)dF(v)}{1\!-\!F(u)}\!\Bigg)\!\frac{dM_i^{c*}(u)}{\widehat{K}_c(u-)}\!\Bigg)^2\nonumber\\&\times\! I\!\Bigg\{\!\Bigg|\!{h}_1(X'_i;t)\!-\!\int\limits_0^\infty\!\Bigg(\!{h}_1(X'_i;t)\!-\!\frac{\int\limits_u^{\infty}\!{h}_1(v;t)dF(v)}{1-F(u)}\!\Bigg)\!\frac{dM_i^{c*}(u)}{\widehat{K}_c(u-)}\!\Bigg|\!>\!\sqrt{n}\!\varepsilon\sigma_*(t)\!\Bigg\}\!\Bigg]\nonumber\\&
   =\lim\limits_{n\to\infty}\!\frac{1}{\sigma_*^2(t)}\!E^*\!\Bigg[\!\Bigg(\!{h}_1(X'_1;t)\!-\!\int\limits_0^\infty\!\Bigg(\!{h}_1(X'_1;t)\!-\!\frac{\int\limits_u^{\infty}\!{h}_1(v;t)dF(v)}{1-F(u)}\!\Bigg)\!\frac{dM_1^{c*}(u)}{\widehat{K}_c(u-)}\!\Bigg)^2\nonumber\\&\times\! I\!\Bigg\{\!\Bigg|\!{h}_1(X'_1;t)\!-\!\int\limits_0^\infty\!\Bigg(\!{h}_1(X'_1;t)\!-\!\frac{\!\int\limits_u^{\infty}\!{h}_1(v;t)dF(v)}{1-F(u)}\!\Bigg)\!\frac{dM_1^{c*}(u)}{\widehat{K}_c(u-)}\!\Bigg|\!>\!\sqrt{n}\varepsilon\sigma_*(t)\!\Bigg\}\!\Bigg]
\end{align}
where $\sigma^2_*(t)=Var^*(Z^*_n(t))$ converges in probability to $\sigma^2=cov(\eta(t),\eta(t))$. The last holds from
 $\widehat{K}_c(x)\stackrel{P}{\to}K_c(x)$, and $\max\limits_{i=1,...,n}\frac{K_c(X_i-)}{\widehat{K}_c(X_i-)}=O_p(1)$ (\cite{zhou1991some}) and $\widehat{\Lambda}(t)\stackrel{P}{\to}\Lambda(t)$.
Further, from the $L^2-$boundness (with respect to measure $dF$) of 
$${h}_1(X'_1;t)\!-\!\int\limits_0^\infty\!\Big(\!{h}_1(X'_1;t)\!-\!\frac{1}{1\!-\!F(u)}\!\int\limits_u^{\infty}\!{h}_1(v;t)dF(v)\!\Big)\!\frac{dM_1^{c*}(u)}{\widehat{K}_c(u-)},$$
 we get that that the expression in the nominator of \eqref{LFCLT} converges to zero.

One  can analogously show that all finite-dimensional distributions $\{Z^*_n(t)\}$ are multivariate normals. 
The proof of the tightness goes along the same lines as for the process $\{Z_n(t)\}.$ This completes the proof of part \ref{b1}.

The rest of the  proof goes in the same way as the proof of the  Theorem \ref{distribution} (parts \ref{korak1} and \ref{korak2} b.).
\end{proof}
\section*{Acknowledgement}
We would like to thank the anonymous referees for their valuable  remarks and suggestions that improved the paper.

\end{document}